\documentclass[10 pt,conference, letterpaper]{IEEEtran}
\usepackage{times,amsmath,epsfig,latexsym,amssymb,psfrag,eufrak,dsfont}
\usepackage{amsfonts}

\usepackage{cite}
\usepackage{booktabs}
\usepackage{paralist}
\usepackage[ruled,vlined]{algorithm2e}
\usepackage{float}
\restylefloat{table}
\usepackage{graphicx}
\usepackage{caption}
\usepackage{subcaption}
\usepackage{float}
\floatstyle{plaintop}
\restylefloat{table}

\usepackage{epsfig}
\usepackage{epstopdf}

\usepackage{color}
\usepackage{stfloats}

\usepackage{subcaption}
\usepackage[utf8]{inputenc}
 
\newenvironment{proof}[1][Proof]{\begin{trivlist}
\item[\hskip \labelsep {\bfseries #1}]}{\end{trivlist}}

\newcommand{\qed}{\nobreak \ifvmode \relax \else
\ifdim\lastskip<1.5em \hskip-\lastskip
\hskip1.5em plus0em minus0.5em \fi \nobreak
\vrule height0.75em width0.5em depth0.25em\fi}
\newtheorem{prop}{Proposition}

\usepackage{upquote}

\def\({\left(}
\def\){\right)}

\def\[{\left[}
\def\]{\right]}
\def\SIR{\textsf{SIR}}
\def\ASE{\textsf{ASE}}

\begin{document}
\graphicspath{{figures/}}
%
\title{On the Asymptotic Behavior of Ultra-Densification under a Bounded Dual-Slope Path Loss Model}


\author{\IEEEauthorblockN{Yanpeng Yang, Jihong Park\IEEEauthorrefmark{2} and Ki Won Sung}
\IEEEauthorblockA{KTH Royal Institute of Technology, Wireless@KTH, Stockholm, Sweden\\
 E-mail: yanpeng@kth.se, sungkw@kth.se}
\IEEEauthorblockA{\IEEEauthorrefmark{2}Dept. of Electronic Systems, Aalborg University, Denmark\\
E-mail: jihong@es.aau.dk}
}

\maketitle

\begin{abstract}
In this paper, we investigate the impact of network densification on the performance in terms of downlink signal-to-interference (SIR) coverage probability and network area spectral efficiency (ASE). A sophisticated bounded dual-slope path loss model and practical user equipment (UE) densities are incorporated in the analysis, which have never been jointly considered before. By using stochastic geometry, we derive an integral expression along with closed-form bounds of the coverage probability and ASE, validated by simulation results. Through these, we provide the asymptotic behavior of ultra-densification. The coverage probability and ASE have non-zero convergence in asymptotic regions unless UE density goes to infinity (full load). Meanwhile, the effect of UE density on the coverage probability is analyzed. The coverage probability will reveal an U-shape for large UE densities due to interference fall into the near-field, but it will keep increasing for low UE densites. Furthermore, our results indicate that the performance is overestimated without applying the bounded dual-slope path loss model. The derived expressions and results in this work pave the way for future network provisioning.
\end{abstract}



%
\IEEEpeerreviewmaketitle

\section{Introduction}

Network densification is considered as a key enabler to cope with the upcoming 5G data tsunami \cite{Bhushan}\cite{Kamel1}. Deploying more base stations (BSs) can rapidly increase the network capacity by shortening the BS and user equipment (UE) association distance as well as by reducing per-cell traffic load. As densification goes on, BS density may easily exceed UE density, forming an ultra dense network (UDN) \cite{Lopez}. Its simplest example can be off-peak traffic hours under dense BS deployment. Peak hours can also be suitable cases since average BS load in practice is only 20\% due to network stability \cite{Orange:14}. In the UDN, a large number of UE-void BSs within their coverages emerge, and the overall network transits from being fully loaded to partially loaded in which not all BSs are \emph{active} (i.e., transmitting signals to serve the UEs within their cells). Such a UDN may evaporate the advantage of densification since there is less than one UE per cell on average. Therefore, it is crucial to understand the asymptotic behavior of ultra densification for the purpose of network deployment. 

The pioneering work \cite{Andrews} provides a comprehensive understanding on the impact of BS density in a fully loaded downlink cellular network with a simple single-slope unbounded path loss model. Illustrated in Fig.~1 as a baseline, it concludes that BS densification does not change the signal-to-interference ratio ($\SIR$) of an individual UE, but linearly improves the area spectral efficiency ($\ASE$) defined as sum rate per unit area. As BS density grows, the desired signal and interference growths cancel each other, leading to such result. However, it is difficult to apply this conclusion to UDNs due to its simplified signal propagation and load models \cite{Ding2}\cite{Liu}. 

\begin{figure}[t]
\hspace{-.3cm}\includegraphics[width=.48\textwidth]{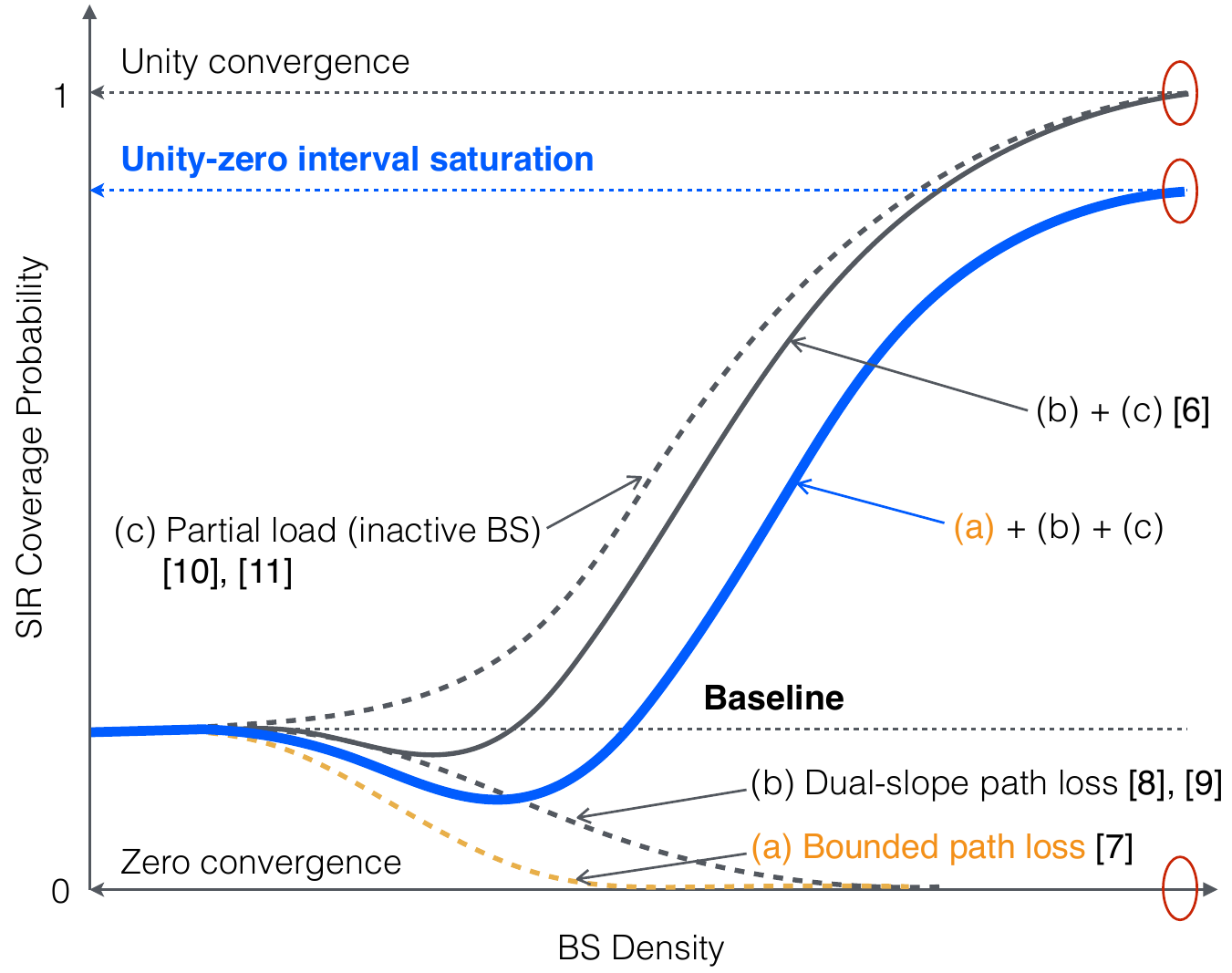}
\caption{Impacts of $(a)$ bounded path loss, $(b)$ dual-slope path loss, and $(c)$ partial load models on asymptotic {\SIR} coverage probability as BS density increases.\protect\footnotemark }
\label{fig:Intro}
\end{figure}

\footnotetext{Minor modification is applied for \cite{Ding} that considers signal-to-interference-plus-noise ratio instead of $\SIR$.}

In a UDN where BS-UE distance $d$ shrinks, a simple unbounded path loss model $d^{-\alpha}$ for the path loss exponent $\alpha>2$ may amplify the received signals when $d<1$, which is unrealistic. In addition, a large amount of signals in a UDN are transmitted from the near-field of a receiver, and these signals experience less attenuation owing to sparse shadowing, i.e. near-field path loss exponent $\alpha_c < \alpha$. A simple single-slope path loss model cannot capture such a distinction. Furthermore, a full load model forces BSs to always transmit signals despite non-negligible portion of UE-void BSs, overestimating interference. It is thus important to incorporate such propagation and load characteristics in detail so as to examine the impact of ultra densification. For this end our system model considers the tri-fold aspect: \emph{$(a)$ bounded path loss, $(b)$ dual-slope path loss and $(c)$ partial load models}, as illustrated in Figs. \ref{fig:Intro} and \ref{fig:system}.

In preceding works, the impacts of $(a)$ and $(b)$ on $\SIR$ coverage probability, defined as $\Pr\(\SIR>t\)$ for a target threshold $t>0$, are respectively investigated by \cite{Liu} and \cite{Zhang,MariosUDN:16}. Both models leads to a conclusion that $\SIR$ coverage probability asymptotically converges to $0$ (i.e. $\SIR \rightarrow 0$) as BS density increases.\begin{footnote}{Near-field environment in \cite{MariosUDN:16} is confined to $\alpha_c<\alpha$ in this article, which originally considers a general $\alpha_c>2$.}\end{footnote}  The reason is as follows: as BS density grows, interference keeps increasing while the desired signal increase is saturated under (a); the number of near-field interferer increases under (b), dominating the increase in the desired signal from `a single' BS. On the other hand, the impact of $(c)$ is clarified in \cite{Park2, Yang}, showing $\SIR$ coverage probability asymptotically converges to $1$, i.e. $\SIR\rightarrow \infty$. It comes from the fact that UEs' neighboring BSs under a nearest association rule are always active while the rest of BSs become inactive. This results in making interfering BS density converge to UE density while the desired signal keeps increasing. A recent work \cite{Ding} combines both $(b)$ and $(c)$, and interestingly concludes that $\SIR$ coverage probability still converges to $1$ asymptotically since $(c)$ dominates $(b)$.

Motivated by the discussions, we aim to combine $(a)$, $(b)$, and $(c)$ altogether, and investigate their aggregate impact on the asymptotic $\SIR$ coverage behavior. The main contributions of this paper are listed below.
\begin{itemize}
\item Asymptotic unity-zero interval $\SIR$ coverage saturation is derived, which also leads to the same ASE saturation. This verifies combining $(a)$ and $(b)$ exactly cancels out $(c)$ (Prop. \ref{prop:converge} and \ref{prop:ASE}).
\item Numerically tractable integral-form of coverage probability under a bounded dual-slope path loss model are derived (\ref{prop:Pc}). Moreover, closed-form bounds of coverage probability and ASE are provided (Prop. \ref{prop:bounds} and \ref{prop:trend}).
\item The impact of UE density on coverage probability and ASE are analyzed (Fig. \ref{fig:PcUE}). Meanwhile, the trend of coverage probability and ASE are interpreted  (Fig. \ref{fig:Pc} and Fig. \ref{fig:ASE}). The scaling trend of ASE in terms of BS density is derived (Prop. \ref{prop:trend}).
\end{itemize}

\section{System Model}

\begin{figure}[t]
\includegraphics[width=.5\textwidth]{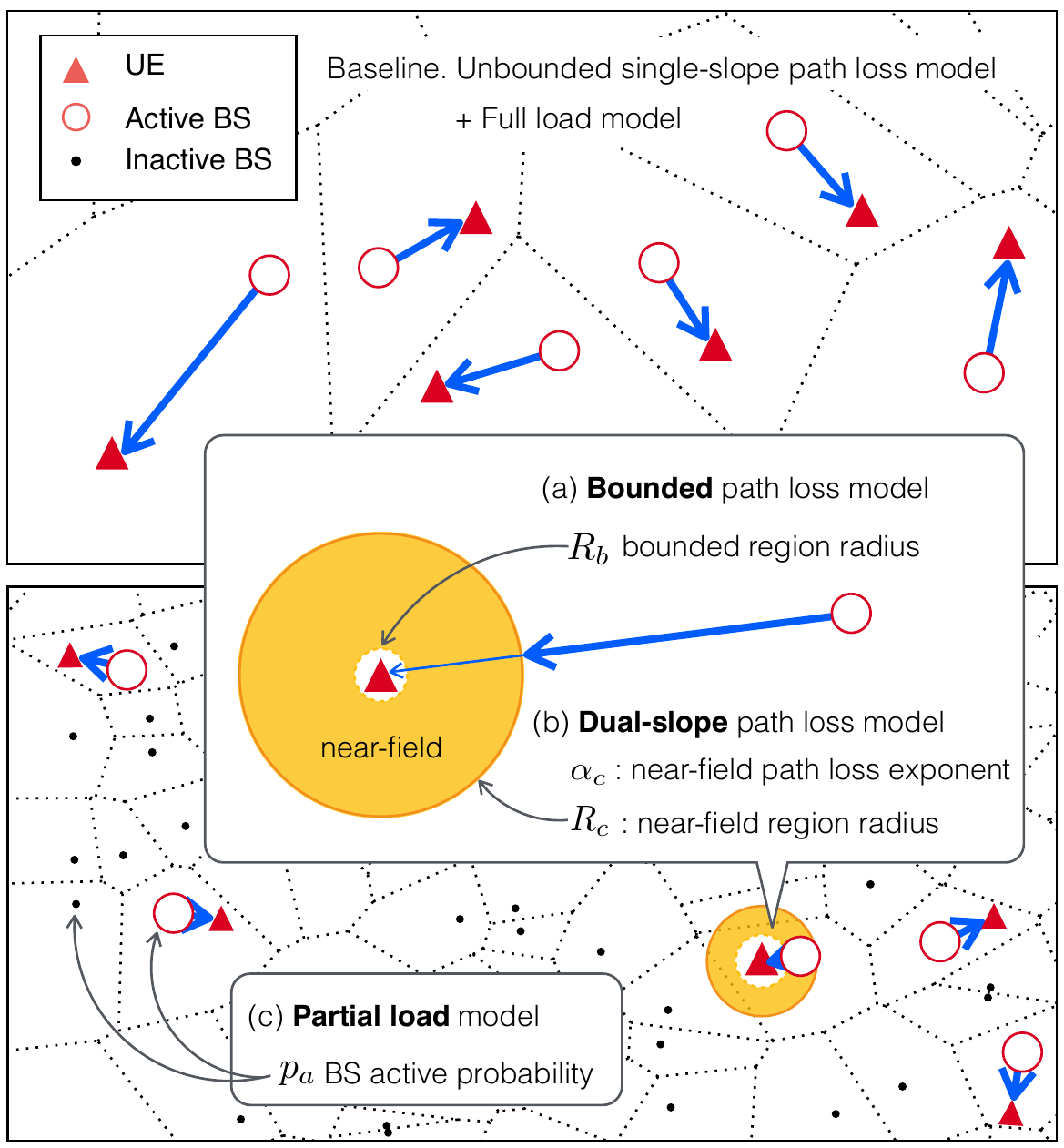}
\caption{Network layout of traditional fully loaded network (top) and partially loaded UDN (bottom) as well as illustration of the bounded dual-slope path loss model (middle).}
\label{fig:system}
\end{figure}

We consider a downlink cellular network where BSs and UEs are distributed according to two independent homogeneous Poisson Point Processes (PPPs) $\Phi_{b}$ and $\Phi_{u}$. The densities of BS and UE are denoted as $\lambda_{b}$ and $\lambda_{u}$ respectively. We assume each UE is associated with its closest BS whose coverage area comprises a Voronoi tessellation as shown in Fig. \ref{fig:system}. Each BS becomes inactive without transmitting any signal when its coverage area, the Voronoi cell, is empty of active UEs. Correspondingly, each active BS has at least one UE in its cell and will randomly choose one of them to serve. According to \cite{Yu}, the probability that a BS becomes active is given as
\begin{align}
p_a \approx 1-\(1+\frac{\lambda_u}{3.5\lambda_b}\)^{-3.5}
\end{align}
which allows us to incorporate the partial load model.

Both BS and UE are equipped with a single antenna and BSs transmit with unit power. Rayleigh fading is used to model the channel gain, with the fading coefficients $h$ are i.i.d zero mean unit variance complex normal distributed random variables. Since we will focus on the asymptotic behavior and the system is interference-limited in dense networks, we will neglect the noise power and examine $\SIR$ throughout the paper.

We consider that path loss attenuation from a BS to a UE is distinguished under two different regions, \emph{near-field} and $\emph{far-field}$ at the UE as illustrated in Fig. \ref{fig:system}. In a near-field within radius ${R_c}$, a transmitted signal experiences less absorption and diffraction, so the near-field path loss exponent $\alpha_c$ becomes less than the far-field exponent $\alpha>2$. In the innermost near field, the transmitted signal becomes no longer attenuated within radius $R_b$ because of the physical volume of the UE. This dual-slope path loss $\ell(\alpha, \alpha_{c}, d)$ can be formulated as a piece-wise function of the propagation distance $d$, shown as
\begin{align}
\label{equ:pathloss}
\ell(\alpha, \alpha_{c}, d)=\left\{
        \begin{array}{ll}
          1, & 0 \leq d \leq {R_{b}}; \\
          d^{-\alpha_{c}}, &  {R_{b}} < d \leq {R_{c}}; \\
          \tau d^{-\alpha}, &  d > {R_{c}}
        \end{array}
      \right.
\end{align}
where ${R_{b}} > 0$ is the radius of bounded path loss region, i.e. the path loss in the range of $[0,{R_b}]$ is assumed constant; $\tau \triangleq {R_{c}}^{\alpha-\alpha_{c}}$; ${R_{c}}  \geq {R_{b}}$ is the critical distance to divide the near- and far-field; and $\alpha_{c}$ and $\alpha$ are the near- and far-field path loss exponents for $2  < \alpha_{c} < \alpha$, respectively.

\begin{figure*}[!b]
\vspace*{4pt}
\hrulefill
\normalsize

\setcounter{equation}{5}
{\small \begin{align}
\label{equ:pcbdm}
P_{c}(\lambda_b,\lambda_u,T) &= \lambda_b \pi \( \int_0^{{R_{b}}^2}e^{-\lambda_{b}\pi r(1+p_{a}G_1(r,T))} \mathrm{d}r + \int_{{R_{b}}^2}^{{R_{c}}^2} e^{-\lambda_{b}\pi r(1+p_{a}G_2(r,T))} \mathrm{d}r + \int_{{R_{c}}^2}^{\infty}e^{-\lambda_{b}\pi r(1+p_{a}G_3(T))} \mathrm{d}r \) \\
\label{equ:G1}
G_1(r,T) &= c_T(\alpha, \alpha_c, {R_c}, {R_b})r^{-1} -T(1+T)^{-1}\\
\label{equ:G2}
G_2(r,T) &=  \[c_T\(\alpha, \alpha_c, {R_c}, \sqrt{r} \) +  {{R_{b}}}^2 \(F\(\frac{2}{\alpha},\frac{1}{T}\)-\frac{T}{1+T}\) \]r^{-1}  - F\(\frac{2}{\alpha},T^{-1}\)\\
\label{equ:G3}
G_3(T) &= T \(\frac{\alpha_c}{2}-1 \)^{-1} F\( 1-\frac{2}{\alpha_c}, T\) \\
c_T(\alpha, \alpha_c, {R_c}, x) &:= {{R_c}}^2 F\(\frac{2}{\alpha},  \frac{1}{T}\[\frac{{R_c}}{x}\]^\alpha\) - {{R_{b}}}^2 \[F\(\frac{2}{\alpha},\frac{1}{T}\)-\frac{T}{1+T}\] + \frac{2 T {R_{b}}^{\alpha} {R_{c}}^{2-\alpha}   }{\alpha_c-2} F\( 1-\frac{2}{\alpha_c}, T\[\frac{x}{{R_c}}\]^\alpha\)
\end{align}
\small \begin{align}
\label{equ:PcLB}
\setcounter{equation}{\value{equation}}
\hspace{-5pt} P_{c}^{\text{LB}}&=\frac{1}{H_1}\(e^{-\lambda_b  p_{a} \pi c_T(\alpha, \alpha_c, {R_c}, {R_b})} - e^{-\lambda_b \pi \({R_b}^2 H_1  + p_{a}c_T(\alpha, \alpha_c, {R_c}, {R_b})\)}\) + \frac{1}{H_{2l}}\(e^{-\lambda_b \pi {R_b}^2 H_{2l}} - e^{-\lambda_b \pi  {R_c}^2 H_{2l}}\) + \frac{1}{H_3}\(e^{-\lambda_b \pi  {R_c}^2 H_3 }\) \\
\label{equ:PcUB}
P_{c}^{\text{UB}}&=\frac{1}{H_1}\(e^{-\lambda_b  p_{a} \pi c_T(\alpha, \alpha_c, {R_c}, {R_b})} - e^{-\lambda_b \pi \({R_b}^2 H_1  + p_{a}c_T(\alpha, \alpha_c, {R_c}, {R_b})\)}\) + \frac{1}{H_{2u}}\(e^{-\lambda_b \pi {R_b}^2 H_{2u}} - e^{-\lambda_b \pi  {R_c}^2 H_{2u}}\) + \frac{1}{H_3}\(e^{-\lambda_b \pi  {R_c}^2 H_3 }\)
\end{align}
}

\end{figure*}

\setcounter{equation}{2}
\subsection{Performance Metrics} 
In this paper, we will focus on two performance metrics from both user and network perspectives: $\SIR$ coverage probability of a typical UE and the $\ASE$ of the network. We analyze the performance of a typical user located at the origin $o$ randomly selected by the BS, which is permissible in a homogeneous PPP by Slivnyak's theorem \cite{Stoyan}. The SIR of a typical user denoted as 0 can be expressed as
\begin{align}
\label{equ:SIR}
\mathrm{SIR}_0=\frac{|h_{0,0}|^2\ell(d_{0,0})}{\displaystyle\sum_{i \in \Phi_{b}^*\backslash \{0\}}|h_{i,0}|^2\ell(d_{i,0})}.
\end{align}
where $d_{i,j}$ and $h_{i,j}$ denote the distance and channel between BS $i$ and UE $j$, $|h_{i,j}|^2 \sim \mathrm{exp}(1)$. $\Phi_{b}^*$ represents the set of active BSs which is not a homogeneous PPP. Nevertheless, we can assume  $\Phi_{b}^*$ as a homogeneous PPP with density $\lambda_b^*=\lambda_b p_a$, which has been shown to be accurate according to \cite{Lee}\cite{Li2}. 
Given the downlink SIR of the typical user, the coverage probability is defined as:
\begin{align}
P_c(\lambda_b, \lambda_u, T) \triangleq \mathbb{P}[\mathrm{SIR}_0>T]
\end{align}
where T is the target $\SIR$ level.

The network $\ASE$ $\Gamma$ is defined as the sum average spectral efficiency of all active BSs achieving the target threshold in a unit area \cite{Li2} and is given by 
\begin{align}
\label{equ:ASE}
\Gamma(\lambda_b, \lambda_u, T) \triangleq p_a \lambda_b   P_c(\lambda_b, \lambda_u, T) \log_2(1+T)
\end{align}
where $p_a\lambda_b P_c$ can be interpreted as the density of the BSs that successfully transmit the symbols to their users.

\section{Coverage Probability and ASE analysis}
\label{sec:FiniteUE}
In this section, we derive the coverage probability and ASE expressions under a bounded dual-slope path loss model and provide closed-form bounds of them. Furthermore, we demonstrate the convergence of them in asymptotic regions where $\lambda_b \to \infty$. 

%

\begin{prop}
\label{prop:Pc}(Coverage probability expression)
In a cellular network with BS active probability $p_a$, the coverage probability under a bounded dual-slope path loss model is expressed in (\ref{equ:pcbdm}) at the bottom of this page, where the supplementary equations are listed in (7)-(10) and $F(b,z)=\,_2F_1(1,b,1+b,-z)$ with $\,_2F_1(a,b,c,z)$ being the Gauss hypergeometric function.
\end{prop} 

\begin{proof}
See Appendix \ref{sec:P1}.
\end{proof}

Despite the complicated form of (\ref{equ:pcbdm}), the first and third integrals can be calculated into exponential expressions. In this case, by applying transforms to the second integral, we can derive closed-form bounds of (\ref{equ:pcbdm}) with only exponential and hypergeometric functions as shown in the following proposition.

\setcounter{equation}{12}
\begin{prop} \label{prop:bounds}(Coverage probability bounds) $\SIR$ coverage probability's lower bound $P_c^{\text{LB}}$ and upper bound $P_c^\text{UB}$ are given as (\ref{equ:PcLB}) and (\ref{equ:PcUB}) at the bottom of this page, where $H_1=1-p_{a}\frac{T}{1+T}$, $H_{2l}=1+p_{a}G_2(R_b^2)$, $H_{2u}=1+p_{a}G_2(R_c^2)$, $H_3=1+p_{a}G_3(T)$.

\begin{proof}
See Appendix \ref{sec:P2}.
\end{proof}
\end{prop}

Applying Proposition 2 in asymptotic regions leads to the following proposition. 
\begin{prop} (Asymptotic $\SIR$ coverage probability)
\label{prop:converge}
As $\lambda_b \rightarrow \infty$, $\SIR$ coverage probability $P_c(\lambda_b, \lambda_u, T)$ converges to a finite value as follows.
\begin{align}
\lim_{\lambda_b\rightarrow\infty} P_c(\lambda_b, \lambda_u, T) &=  e^{-\lambda_u \pi c_T(\alpha, \alpha_c, {R_c}, {R_b})}
\end{align} 
\end{prop} 
\begin{proof}
From (\ref{equ:PcLB}) and (\ref{equ:PcUB}) in Proposition \ref{prop:Pc}, we have
$\lim_{\lambda_b \to \infty} P_{c}^{\text{LB}} = \lim_{\lambda_b \to \infty} P_{c}^{\text{UB}} = \frac{1}{H_1}e^{-\lambda_u \pi c_T(\alpha, \alpha_c, {R_c}, {R_b})}$ since all the other terms tend to 0 as $\lambda_b \to \infty$. According to the Squeeze theorem, $\lim_{\lambda_b \to \infty} P_{c} = \frac{1}{H_1}e^{-\lambda_u \pi m}$. Meanwhile, $\lim_{\lambda_b \to \infty} H_1 = 1$ since $p_a \to 0$. Thus $\lim_{\lambda_b \to \infty} P_{c} = e^{-\lambda_u \pi c_T(\alpha, \alpha_c, {R_c}, {R_b})}$.
\end{proof}

Proposition \ref{prop:converge} emphasizes the importance of considering UE density in a UDN. The converged value is a decreasing function of $\lambda_u$ and it tends to 0 when $\lambda_u \to \infty$, i.e., in a fully loaded network. Thus, deploying infinite number of BSs will not bring the UE performance to a unprecedented level. In contrast, extreme densification will put the coverage probability into the danger of decreasing to 0, as shown in Fig. \ref{fig:PcUE} in the next section. The converged result also depends on environmental parameters $(\alpha, \alpha_c, {R_c}, {R_b})$. It will increase as both path loss exponents grow since the coverage probability is a increasing function of path loss exponents \cite{Zhang}. A larger $R_b$ or $R_c$ will decline the performance since it either reduce the signal power or amplify the interference in the \emph{near-field}.

We now turn to the network perspective and study the asymptotic behavior of the ASE. Combining the definition in (\ref{equ:ASE}) with Proposition \ref{prop:converge}, we can easily obtain the following proposition.
\begin{prop} (Asymptotic $\ASE$)  As $\lambda_b \rightarrow \infty$, $\ASE$ $\Gamma$ converges to a finite value as follows.

\vspace{-10pt}\small \begin{align}
\lim_{\lambda_b \rightarrow \infty} \Gamma\(\lambda_b, \lambda_u, T\) &= \lambda_u  e^{-\lambda_u \pi c_T(\alpha, \alpha_c, {R_c}, {R_b})}\log_2(1 + T)
\end{align}\normalsize 
\label{prop:ASE}
\end{prop}

Proposition \ref{prop:ASE} shows that the asymptotic ASE will increase with $\lambda_u$ when $\lambda_u$ is small and tends to 0 as $\lambda_u \to \infty$. Unlike the coverage probability, the asymptotic ASE will be beneficial from ultra-desification to some extent and but still highly depends on the UE density. 

Returning from the asymptotic regions, we demonstrate how the ASE scales with BS density in the next proposition.
\begin{prop} \label{prop:trend}($\ASE$ scaling)
The $\ASE$ scales with $\lambda_b p_a e^{-\lambda_b p_a \pi c_T(\alpha, \alpha_c, {R_c}, {R_b})}$ and is bounded by
\begin{align}
\Gamma^{\mathrm{LB}} = \lambda_b p_a P_{c}^{\text{LB}} \log_2(1+T)
\end{align}
\begin{align}
\Gamma^{\mathrm{UB}} = \lambda_b p_a P_{c}^{\text{UB}} \log_2(1+T).
\end{align}
\end{prop}
\begin{proof}
See Appendix \ref{sec:P3}.
\end{proof}

Similar with its asymptotic behavior, the ASE will increase with BS density when $\lambda_b$ is small and finally converge to $\lambda_u e^{-\lambda_u \pi c_T} \log_2(1+T)$ as shown in Proposition \ref{prop:ASE}.

\begin{figure}[t]
\includegraphics[width=.5\textwidth]{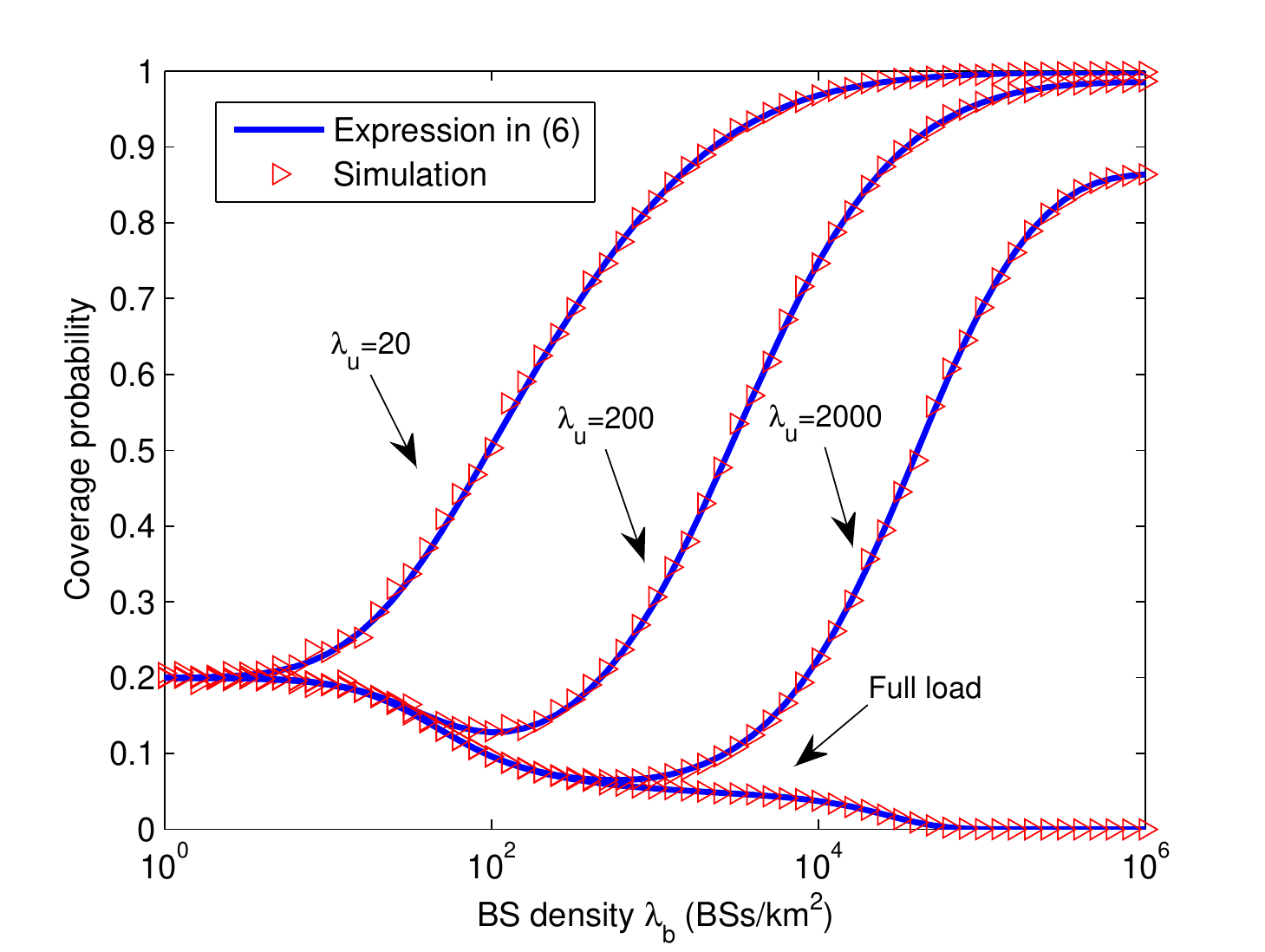}
\caption{Effect of $\lambda_u$ on coverage probability under bounded dual-slope model.}
\label{fig:PcUE}
\end{figure}


\section{Numerical Results}
In this section, we present the numerical results to study the performance of network densification and validate our theoretical analysis. We assume ${R_b}=1m$, ${R_c}=70m$, $\alpha_c=2.5$, $\alpha=4$ and set the SIR threshold T = 10dB in all of our results. To calculate or simulate `fully loaded network', we set $\lambda_u=2\times10^8/\text{km}^2$ which is a sufficiently large value so that $p_a \approx 1$. 
\subsection{Effect of UE density}
Fig. \ref{fig:PcUE} shows the effect of UE density on coverage probability. An exact match between simulation and analysis is observed. Meanwhile, we find that coverage probabilities show completely different trends among different UE densities.  

\begin{figure}[t]
  \centering
  \begin{subfigure}[a]{0.5\textwidth}
    \centering\includegraphics[width=0.9\textwidth]{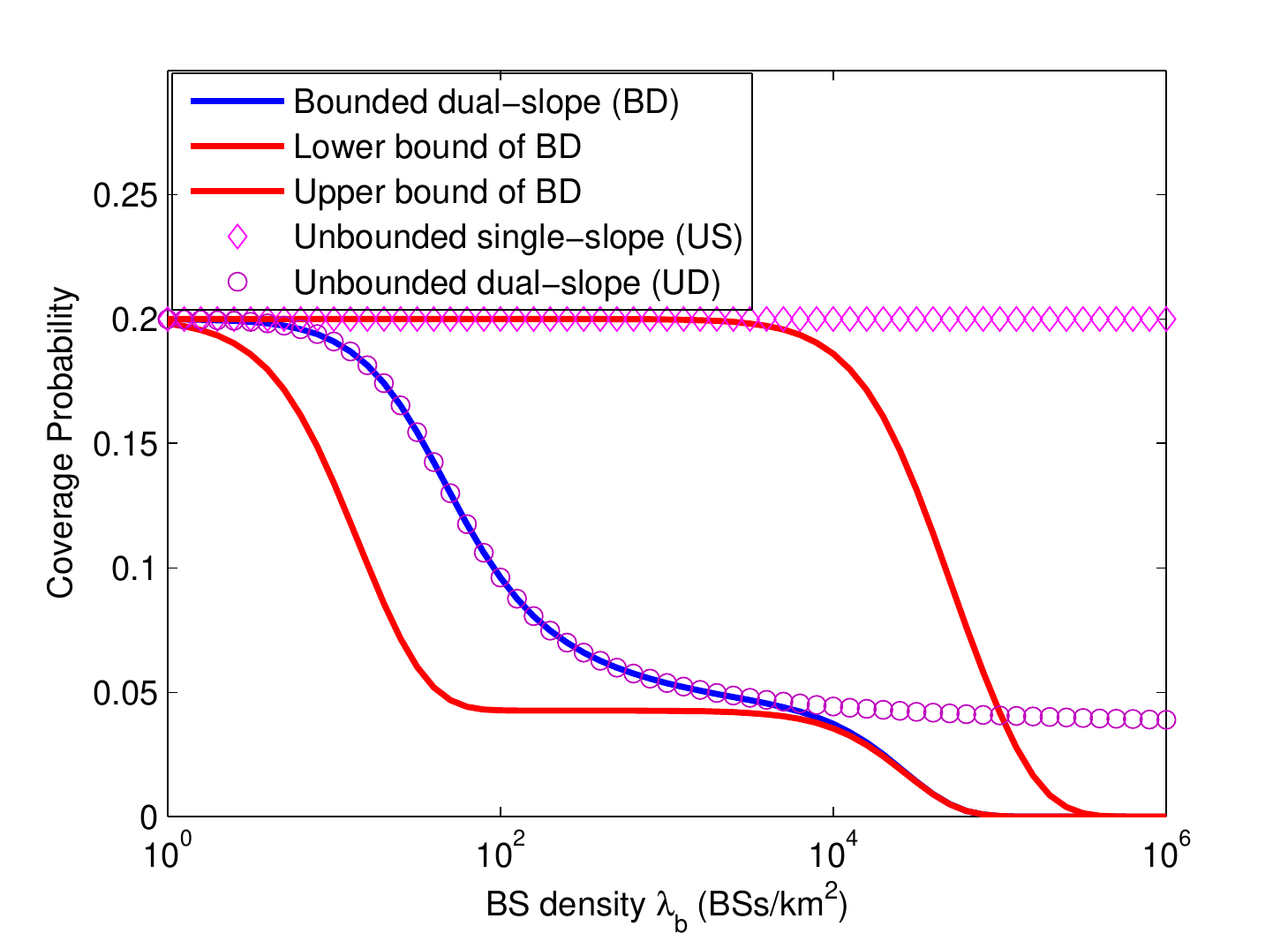}
    \caption{Coverage Probability when network is fully loaded}
    \label{fig:PcInf}
  \end{subfigure}\\
  \begin{subfigure}[a]{0.5\textwidth}
    \centering\includegraphics[width=0.9\textwidth]{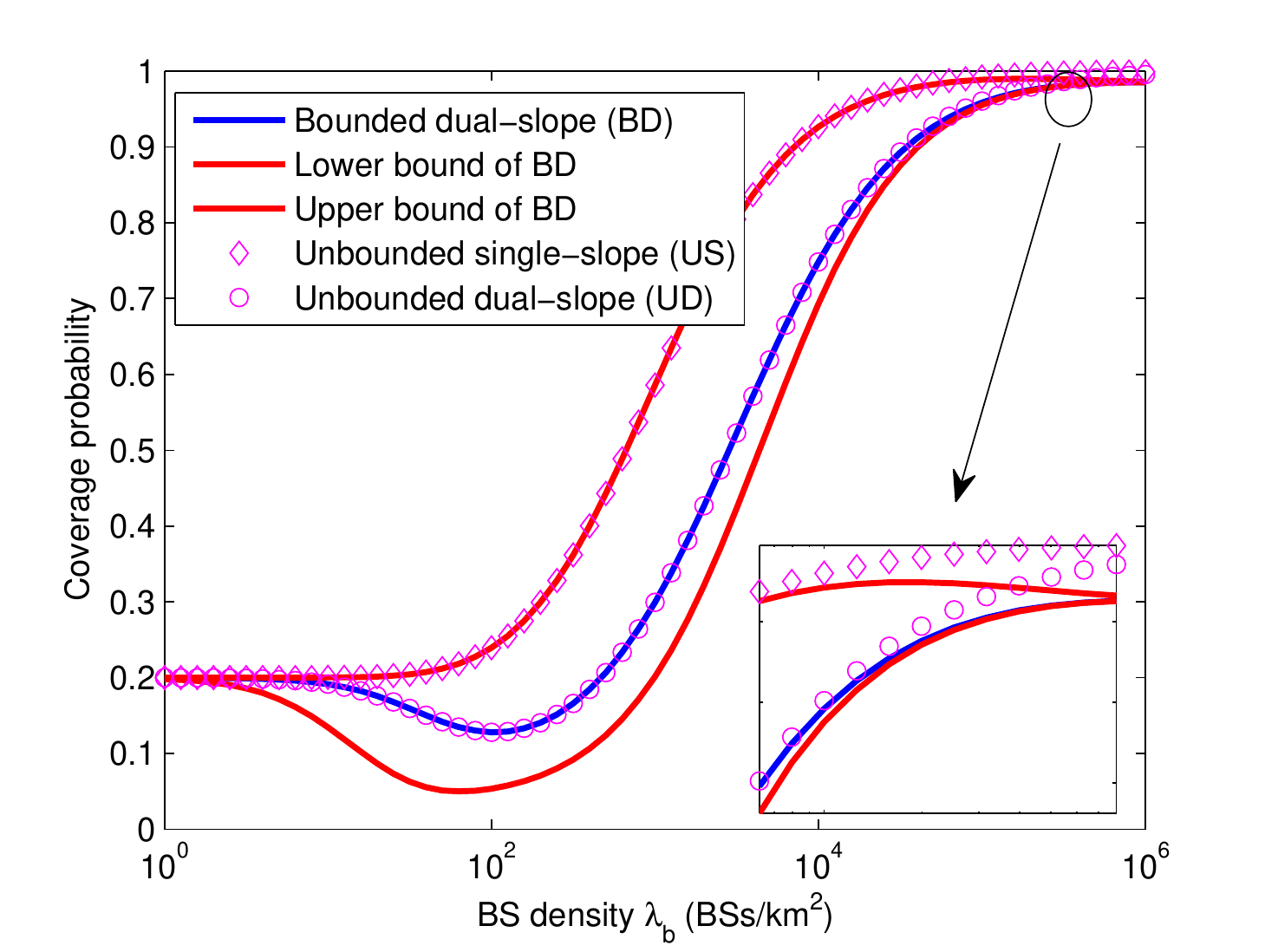}
    \caption{Coverage Probability when $\lambda_u = 200$}
    \label{fig:Pc200}
  \end{subfigure}
  \caption{Coverage probability bounds and comparison with previous models}
   \label{fig:Pc}   
\end{figure}

In full load model, the diminishing of coverage probability starts when interfering BSs fall into the \emph{near-field} of the typical UE and keeps decreasing since interference will continue increasing. When $\lambda_u$ is finite, the interferer coordinates converge to UE coordinates in a UDN regime \cite{Park2}. Thus the distance from the typical UE to its closest interferer can be approximated as the distance to the its closest neighbor UE , which has an expected value of $\frac{1}{2\sqrt{\lambda_u}}$. When UE density is low (e.g. $\lambda_u=20$), the expected value is larger than the critical distance, which means the probability of no interferer inside the \emph{near-field} of the typical UE is very high. Hence, the coverage probability is a non-decreasing function of $\lambda_b$ as in a single-slope model. In contrast, higher UE density (e.g. $\lambda_u=200$ or $2000$) leads to more potential interferers within critical distance. Thus coverage probability will decrease for the same reason as in fully loaded network. Nevertheless, when all the UEs in the near-field get service, coverage probability will start increasing again since the interference are saturated and no longer increase. Therefore, it is important to estimate the active UE density for efficient network deployment or operation in order to avoid the decreasing region of coverage probability.

\subsection{Coverage probability analysis}

In Fig. \ref{fig:Pc}, we compare the coverage probability under our model with the previous models. We observe that the performance is overestimated with unbounded or single-slope models in highly densified regions. The reason is that those models either exaggerate the received power inside the bounded region or underestimate the interference in the near-field. 

The inaccuracy of path loss models may mislead the prediction of asymptotic behavior. For instance, the coverage probability will converge to 1 with unbounded models but to $e^{-\lambda_u \pi c_T(\alpha, \alpha_c, {R_c}, {R_b})}$ which is smaller than 1 (assume $\lambda_u > 0$) when applying a bounded model. Consistent with the result in Proposition \ref{prop:converge}, the converged value will decrease as UE density grows and finally falls to zero in the full load case as shown in Fig. \ref{fig:PcUE}. This is because the signal is limited by the bound effect and the overall performance will be dominated by the interference which depends on UE density.

\begin{figure}[t]
  \centering
  \begin{subfigure}[a]{0.5\textwidth}
    \centering\includegraphics[width=0.9\textwidth]{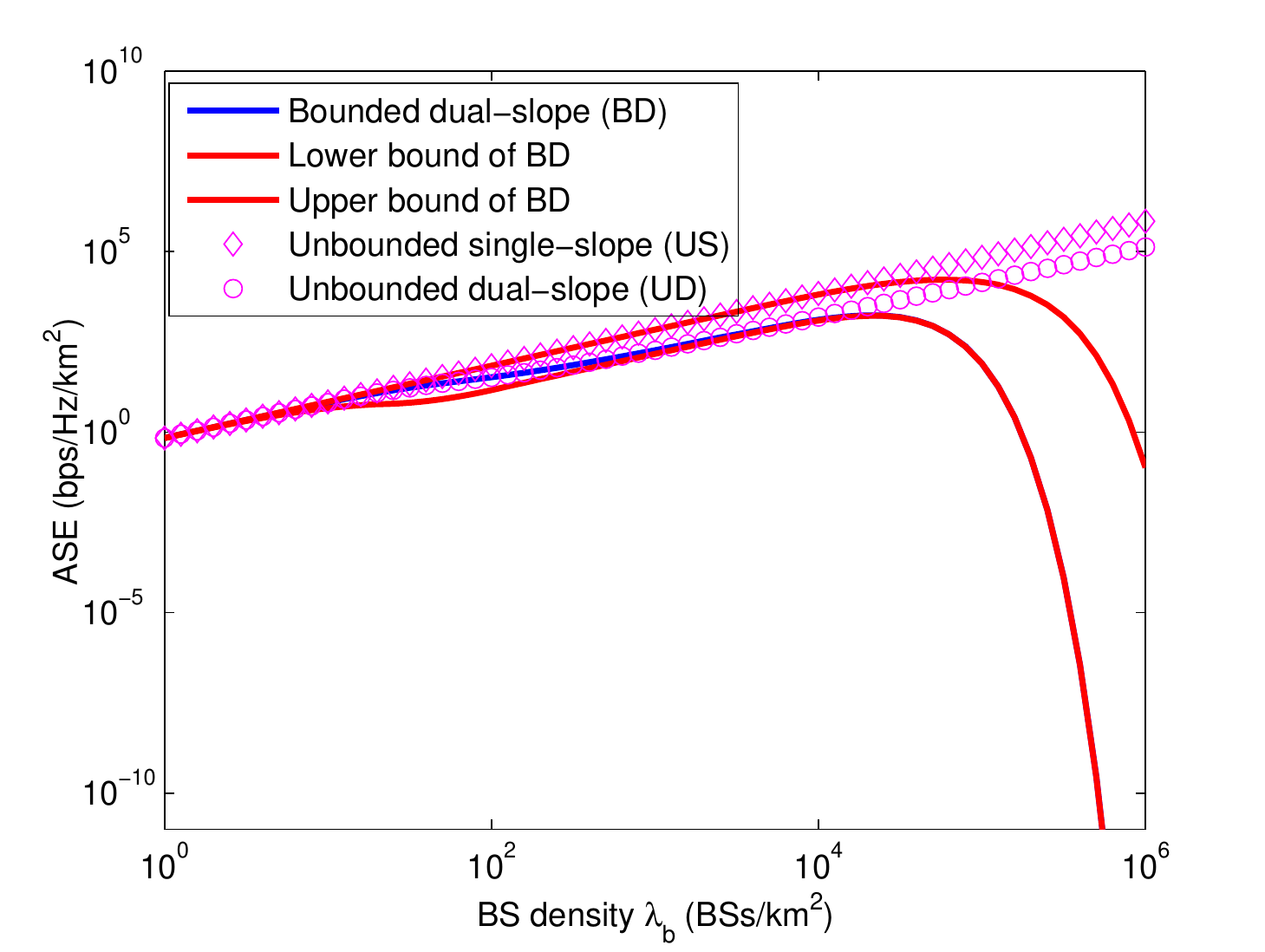}
    \caption{ASE when network is fully loaded}
    \label{fig:ASEInf}
  \end{subfigure}\\
  \begin{subfigure}[a]{0.5\textwidth}
    \centering\includegraphics[width=0.9\textwidth]{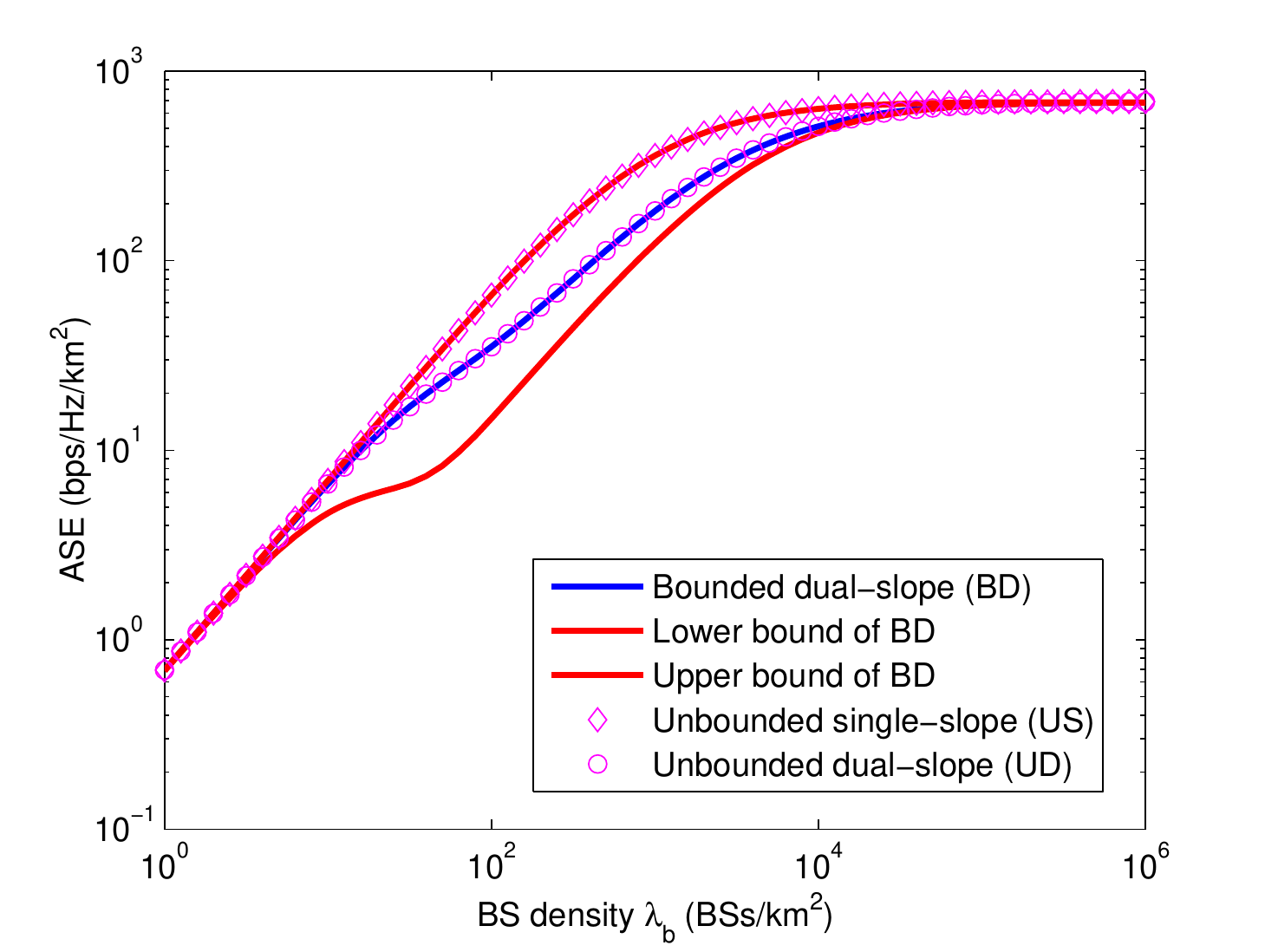}
    \caption{ASE when $\lambda_u = 200$}
    \label{fig:ASE200}
  \end{subfigure}
  \caption{ASE bounds and comparison with previous models.}
   \label{fig:ASE}   
\end{figure}

\subsection{ASE analysis}
Figure \ref{fig:ASE} depicts the scaling of ASE with regard to BS density. Aligning with proposition \ref{prop:trend}, the ASE first increase with BS density and then converge to a constant. The constant is larger than 0 in partially loaded network and decreases to 0 when the network is full load ($\lambda_u \to \infty$) as proved in Proposition \ref{prop:ASE}. 

By comparing Fig. \ref{fig:ASE} with Fig. \ref{fig:Pc}, we can observe a trade-off between UE and network performance during BS densification. In full load case, there exists an BS density threshold around $\lambda_b = 10^4$. Before the threshold, although individual performance gets worse, densification is still beneficial from the network perspective. For partially loaded network, the trade-off appears approximately between $\lambda_b = 10^{1.5}$ to $\lambda_b = 10^3$. The phenomenon further demonstrates the necessity of applying a dual-slope model because the coverage probability is a non-decreasing function of BS density in a single-slope model.

Our bounds in Proposition \ref{prop:bounds} and \ref{prop:trend} are compared with the integral expression and shown in Fig. \ref{fig:Pc} and \ref{fig:ASE}. The figures verify the asymptotic tightness of the bounds for $\lambda \rightarrow \infty$. The upper bound is tighter in small UE density scenarios while the lower bound fits better for large UE densities. This is because the upper bound is close to single-slope model which is similar with small UE density scenario. The bounds can be used as approximations in large BS density regions.

\section{Conclusion}
\label{sec:Con}
In this paper, we investigate the asymptotic behavior of ultra-densification of base stations. To our best knowledge, this is the first work incorporating two key aspects of UDN modeling: a partially loaded network due to a finite active UE density and a dual-slope path loss model with a bounded loss within a unit distance. With such models, we find that the asymptotic behavior of ultra-dense base station deployment is different from what was known with simpler assumptions, e.g. unit or zero convergence of coverage probability. Depending on the UE density, both UE coverage probability and ASE converge to either zero or a constant value. Even before the asymptotic regions, our results suggest that the densification cannot always improve the individual UE performance or boost the network throughput as well. The increment are prevented by introducing extra interference in the near-field until all the UEs in the near-field are served. Our work provides insights into the scaling of the network densification, and thus gives a guideline for the network deployment. 





%
\appendices
\section{Proof of Proposition 1}
\label{sec:P1}
We start from the coverage probability expression under general path loss model and then plug in with our bounded dual-slope model. According to the definition, the coverage probability can be expressed as:
\begin{align}
\begin{split}
\begin{aligned}
\label{equ:PcO}
P_{c}^{l}&=\mathbb{P}[\mathrm{SIR}>T]=\mathbb{P}\[\frac{|h|^2\ell(r)}{I}>T\]\\
&\stackrel{(a)}{=}\int_{r>0}\mathbb{P}\[|h|^2>\frac{TI}{\ell(r)}|r\]f_r(r) \mathrm{d}r\\
&\stackrel{(b)}{=}\int_{r>0}\mathcal{L}_{I}\(\frac{T}{\ell(r)}\)f_r(r)\mathrm{d}r
\end{aligned}
\end{split}
\end{align}
where (a) follows from BS distribution and (b) is due to the fact that $|h|^2 \sim \mathrm{exp}(1)$, $\mathcal{L}_{I}(s)$ is the Laplace transform of interference which can be derived as
\begin{align}
\begin{split}
\begin{aligned}
\mathcal{L}_{I}(s)&=\mathbb{E}_{I}[e^{-sI}]=\mathbb{E}_{\Phi_{b}^{*},g_{i}}[exp(-s(\sum_{x\in\Phi_{b}^{*}}g_{i}\ell(d_{i})))]\\
&\stackrel{(a)}{=}\mathbb{E}_{\Phi_{b}^{*}}\[\prod_{x\in\Phi_{b}^{*}}\frac{1}{1+s\ell(d_i)}\]\\
&\stackrel{(b)}{=}exp\(-2\pi\lambda_{b}^{*}\int_{r}^{\infty}\(1-\frac{1}{1+s\ell(v)}\)\mathrm{d}v\)
\end{aligned}
\end{split}
\end{align}
where (a) is because $g \sim \mathrm{exp}(1)$ and (b) follows the probability generating functional (PGFL) of the PPP. Plugging in $s=\(\frac{T}{\ell(r)}\)$ and employing a change of variables $v=\sqrt{t}r$ results in
\begin{align}
\label{equ:LIT}
\mathcal{L}_{I}\(\frac{T}{\ell(r)}\)=exp\(-2\pi\lambda_{b}P_a\int_{1}^{\infty}\(\frac{T}{T+\frac{\ell(r)}{\ell(\sqrt{t}r)}}\)\mathrm{d}t\).
\end{align}
Plugging (\ref{equ:LIT}) into (\ref{equ:PcO}) with $z \to r^2$ gives the coverage probability under a general path loss fucntion in (\ref{equ:PcL}) as:

{\footnotesize
\begin{align}
\hspace{-10pt}\label{equ:PcL}
P_{c}^{l}(\lambda_b,\lambda_u,T) =
\lambda_b \pi \int_0^{\infty}\exp\(-\lambda_{b}\pi z\[1+p_{a}\int_1^{\infty}\frac{1}{1+\frac{\ell(\sqrt{z})}{T\ell(\sqrt{tz})}}\mathrm{d}t\]\)\mathrm{d}z
\end{align}
}
Based on (\ref{equ:PcL}), we can substitute our bounded dual-slope model (\ref{equ:pathloss}) into it and get the expression in (\ref{equ:pcbdm}). 
\section{Proof of Proposition 2}
\label{sec:P2}
According to the expression of (\ref{equ:G1}) and (\ref{equ:G3}), $rG_1(r,T)$ and $rG_3(T)$ are linear functions of $r$. Thus we can rewrite the first and third integral in (\ref{equ:pcbdm}) as follows:
{\small
\begin{align}
\int_0^{{R_{b}}^2}e^{-\lambda_{b}\pi r(1+p_{a}G_1(r))} \mathrm{d}r = \frac{1}{H_1}\(e^{-\lambda_b  p_{a} \pi c_T} - e^{-\lambda_b \pi \({R_b}^2 H_1  + p_{a}c_T\)}\)
\end{align}
\begin{align}
 \int_{{R_{c}}^2}^{\infty}e^{-\lambda_{b}\pi r(1+p_{a}G_3(r))} \mathrm{d}r = \frac{1}{H_3}\(e^{-\lambda_b \pi  {R_c}^2 H_3 }\).
\end{align}
}
In the second integral, from $G_2'(r)<0$ we can get $G_2({R_b}^2) \geq G(r) \geq G_2({R_c}^2)$. With the inequality, we can provide bounds for the second integral as:
{\small
\begin{align}
\begin{split}
\begin{aligned}
\label{equ:G2u}
\int_{{R_{b}}^2}^{{R_{c}}^2} e^{-\lambda_{b}\pi r(1+p_{a}G_2(r))} \mathrm{d}r & \leq \int_{{R_{b}}^2}^{{R_{c}}^2} e^{-\lambda_{b}\pi r(1+p_{a}G_2({R_b}^2))} \mathrm{d}r \\
&=\frac{1}{H_{2u}}(e^{-\lambda_b \pi H_{2u} {R_b}^2} - e^{-\lambda_b \pi H_{2u} {R_c}^2})\\
\end{aligned}
\end{split}
\end{align}
\begin{align}
\begin{split}
\begin{aligned}
\label{equ:G2l}
\int_{{R_{b}}^2}^{{R_{c}}^2} e^{-\lambda_{b}\pi r(1+p_{a}G_2(r))} \mathrm{d}r &\geq \int_{{R_{b}}^2}^{{R_{c}}^2} e^{-\lambda_{b}\pi r(1+p_{a}G_2({R_c}^2))} \mathrm{d}r \\
&=\frac{1}{H_{2l}}(e^{-\lambda_b \pi H_{2l} {R_b}^2} - e^{-\lambda_b \pi H_{2l} {R_c}^2}).
\end{aligned}
\end{split}
\end{align}
}
Replacing the integrals in (\ref{equ:pcbdm}) with the exponential expressions above completes the proof.

\section{Proof of Proposition 3}
\label{sec:P3}
\textbf{Notation}: Let $f$ and $g$ be two functions defined on some subset of the real numbers. One writes $f(x) = \mathcal{O}(g(x))$ if and only if there exists a positive real number $M$ and a real number $x_0$ such that $f(x) \leq Mg(x)$ for all $x \geq x_0$.

We omit the proof of the bounds since they come directly from Proposition \ref{prop:bounds}. To prove the ASE scales with $\lambda_b^* e^{-\lambda_b^* \pi c_T}$ is equivalent with showing $\Gamma^{\mathrm{UB}} = \mathcal{O}(\lambda_b^* e^{-\lambda_b^* \pi c_T})$ and $\lambda_b^* e^{-\lambda_b^* \pi c_T} = \mathcal{O}(\Gamma^{\mathrm{LB}})$. Denote $\log_2(1+T)$ as $\tau$ and from (\ref{equ:PcUB}) we have:
{\small
\begin{align}
\Gamma^{\mathrm{UB}} \leq \lambda_b^* (\frac{1}{H_1} e^{-\lambda_b p_{a} \pi c_T} + \frac{1}{H_{2u}}e^{-\lambda_b \pi H_{2u} {R_b}^2} + \frac{1}{H_3}e^{-\lambda_b \pi H_3 {R_c}^2}) \tau.
\end{align}
}
Then we can show $\exists \lambda_1>0, \forall \lambda_b> \lambda_1,  \frac{1}{H_1} e^{-\lambda_b p_{a} \pi c_T} >  \frac{1}{H_{2u}}e^{-\lambda_b \pi H_{2u} {R_b}^2}$ and $\frac{1}{H_1} e^{-\lambda_b p_{a} \pi c_T} > \frac{1}{H_3}e^{-\lambda_b \pi H_3 {R_c}^2}$ since $\frac{1}{H_1} e^{-\lambda_b p_{a} \pi c_T} \rightarrow \frac{1}{H_1} e^{-\lambda_u \pi c_T}$ and the other two parts $\rightarrow$ 0 as $\lambda_b \rightarrow \infty$. Thus $\exists \lambda_1>0, \forall \lambda_b> \lambda_1, \Gamma^{\mathrm{UB}} \leq \frac{3}{H_1} \lambda_b^* e^{- \lambda_b^* \pi c_T} \implies \Gamma^{\mathrm{UB}} = \mathcal{O}(\lambda_b^* e^{-\lambda_b^* \pi c_T})$.

For $\Gamma^{\mathrm{LB}}$, from (\ref{equ:PcLB}) we have
\begin{align}
\Gamma^{\mathrm{LB}} \geq \frac{1}{k_1} \lambda_b^* \(e^{-\lambda_b p_{a} \pi c_T} - e^{-\lambda_b \pi (H_1 {R_b}^2 + p_{a}c_T)}\) \tau,
\end{align}
which can be rephrased as:
\begin{align}
\lambda_b^* e^{-\lambda_b p_{a} \pi c_T} \leq  \frac{H_1}{(1 - e^{-\lambda_b \pi H_1 {R_b}^2}) \tau} \Gamma^{\mathrm{LB}}.
\end{align}
Thus, $\exists \lambda_2>0, k>0, \forall \lambda_b > \lambda_2, e^{-\lambda_b \pi H_1 {R_b}^2} < k$ thus $\lambda_b^* e^{-\lambda_b p_{a} \pi c_T} \leq  \frac{H_1}{(1 - k) \tau} \Gamma^{\mathrm{LB}} $. Therefore $\lambda_b^* e^{-\lambda_b^* \pi c_T} = \mathcal{O}(\Gamma^{\mathrm{LB}})$ and we complete the proof.
\ifCLASSOPTIONcaptionsoff
  \newpage
\fi

\section*{Acknowledgment}
Part of this work has been supported by the H2020 project METIS-II co-funded by the EU. The views expressed are those of the authors and do not necessarily represent the project. The consortium is not liable for any use that may be made of any of the information contained therein.
  
\bibliographystyle{IEEEtran}
\bibliography{IEEEabrv,bibl}

\end{document}